\documentclass{article}
\usepackage{amsmath,graphicx}
\usepackage{cite}
\usepackage{array}
\usepackage{amssymb,amsmath}		
\usepackage{graphics}		
\usepackage{longtable}          
\usepackage{subfigure}
\usepackage{float}
\usepackage{amssymb,amsmath}

\newcommand{\vectornorm}[1]{{\left|\left|#1\right|\right|}_{\ell_2}} 
 
\newcommand{\vectornormZero}[1]{{\left|\left|#1\right|\right|}_{\ell_0}} 
\newcommand{\vectornormInf}[1]{{\left|\left|#1\right|\right|}_{\ell_\infty}} 
\newtheorem{theorem}{Theorem}[section]
\newenvironment{proof}[1][Proof]{\begin{trivlist}
\item[\hskip \labelsep {\bfseries #1}]}{$\blacksquare$ \end{trivlist}}

\begin{document}

\title{One-Step Quantized Network Coding for Near Sparse Gaussian Messages}
\author{Mahdy Nabaee and Fabrice Labeau
\thanks{M. Nabaee and F. Labeau are with the Department
of Electrical and Computer Engineering, McGill University, Montreal.}
\thanks{This work was supported by Hydro-Québec, the Natural Sciences and Engineering Research Council of Canada and McGill University in the framework of the NSERC/Hydro-Québec/McGill Industrial Research Chair in Interactive Information Infrastructure for the Power Grid.}
}
\maketitle

\begin{abstract}
In this paper, mathematical bases for non-adaptive joint source network coding of correlated messages in a Bayesian scenario are studied.
Specifically, we introduce one-step Quantized Network Coding (QNC), which is a hybrid combination of network coding and packet forwarding for transmission.
Motivated by the work on Bayesian compressed sensing, we derive theoretical guarantees on robust recovery in a one-step QNC scenario.
Our mathematical derivations for Gaussian messages express the opportunity of distributed compression by using one-step QNC, as a simplified version of QNC scenario.
Our simulation results show an improvement in terms of quality-delay performance over routing based packet forwarding.
\end{abstract}

\section{Introduction}
\label{sec:Intro}

Unlike routing based \cite{al2004routing} packet forwarding, network coding \cite{ahlswede2000network} offers advantages in terms of flexibility and robustness to deployment changes, for data gathering in sensor networks.
Beyond that, in noisy networks, resistance to link failures and errors has made network coding a better alternative for packet forwarding, as a result of flow diversity in the network \cite{lim2011noisy}.
In this paper, we aim to exploit the possibility of \textit{non-adaptive} distributed source coding by introducing a hybrid network coding and packet forwarding transmission method for correlated sensed data (referred as \textit{messages} in this paper).

Compressed sensing concepts have been recently used in the networks to develop non-adaptive approaches for efficient communication \cite{rabbat,feizi2011power}.
Distributed source coding \cite{xiong2004distributed}, embedded in network coding is recently studied in \cite{naba1,bassi2012compressive}, where the compressed sensing \cite{CS} is joined with network coding.
Specifically, in \cite{naba1}, we have proposed and formulated \textit{Quantized Network Coding} (QNC), which incorporates random linear network coding in real field and quantization to cope with the limited capacity of the links. 
By using restricted isometry property, theoretical guarantees for robust $\ell_1$-min recovery \cite{naba1} in non-Bayesian QNC scenario was discussed in \cite{naba2}.
In Bayesian scenarios (where a priori is known beyond the sparsity domain), to tackle the design of a practical and near optimal decoder, we proposed a Belief Propagation (BP) based Minimum Mean Square Error (MMSE) decoder \cite{naba3}.
In this paper, we provide theoretical bases to justify adopting Bayesian quantized network coding for correlated Gaussian messages.
Specifically, we derive the required number of received packets for robust recovery of messages, implying an embedded distributed compression of messages, while keeping other advantages of network coding over packet forwarding.

In section~\ref{sec:ProbDesc}, we describe the assumptions on the random network deployment and the used model on near sparse messages.
Then, in section~\ref{sec:QNC}, we introduce and formulate the one-step QNC, and derive theoretical guarantees for robust recovery of near sparse messages.
Our simulation results are presented in section~\ref{sec:simRes}, which is followed by our conclusions in section~\ref{sec:Conclusions}.

\section{Problem Description and Notation}
\label{sec:ProbDesc}
Consider a directed network (graph), $\mathcal{G}=(\mathcal{V},\mathcal{E})$, where $\mathcal{V}=\{1,\cdots,n\}$, and $\mathcal{E}=\{1,\cdots,|\mathcal{E}|\}$, are the sets of nodes and edges (links), respectively.
Each edge, $e$, maintains a lossless communication from $tail(e)$ to $head(e)$, at a maximum rate of $C_e$ bits per use.
This implies the same input and output contents for edge $e \in \mathcal{E}$, denoted by $Y_e(t)$, where $t$ is the time index, representing transmission of a block of length $L$ over the edges.
The edges are uniformly distributed between the pairs of nodes. Explicitly, for each $e \in \mathcal{E}$, we have:
\begin{equation}
\textbf{P}\Big(tail(e)=v,head(e)=v'\Big)=\frac{1}{n(n-1)},~\forall v,v' \in \mathcal{V},
\end{equation}
where $v \neq v'$.
The sets of incoming and outgoing edges of node $v$, are also defined as:
\begin{eqnarray}
In(v)&=&\{e:head(e)=v\}, \\
Out(v)&=&\{e:tail(e)=v\}.
\end{eqnarray}

Each node has a random information source (message) $X_v$, where $X_v \in \mathbb{R}$.
These random messages, $\underline{X}=[X_v:v \in \mathcal{V}]$, are correlated and their correlation can be modeled by using an orthonormal transform matrix, $\phi_{n \times n}$, for which $\underline{S}=\phi^T \underline{X}$ is almost $k$-sparse: $\vectornormZero{\underline{S}} \simeq k$.
In this paper, we assume $S_v$'s are drawn from a two-state Gaussian mixture model, such that there is a random state $Q_v$, for which:
\begin{eqnarray}
Q_v=1 & \rightarrow & S_v \sim \mathcal{N}(0,\sigma^2_s), \nonumber\\
Q_v=0 &\rightarrow & S_v \sim \mathcal{N}(0,\sigma^2_z), \label{Eq:modelMess}
\end{eqnarray}
and $\sigma^2_s$ and $\sigma^2_{z}$, with $\sigma^2_{s} \gg \sigma^2_z$, are the variances of large and small elements of near sparse domain.

In the described noiseless network, we need to transmit messages, $X_v$'s, to a single gateway (decoder) node, $v_0 \in \mathcal{V}$ and recover them with a small or no distortion.
Performing distributed source coding \cite{xiong2004distributed} along packet forwarding is used to transmit correlated messages.
But, dealing with the aforementioned network coding advantageous, we are interested to establish theoretical guarantees for joint source network coding, in the non-adaptive QNC scenario.

\section{One-Step Quantized Network Coding}
\label{sec:QNC}

We have defined and formulated QNC and its Bayesian counterpart in \cite{naba1,naba3}.
In this section, we describe a special case of QNC, in which QNC is done only at the first time instance, \textit{i.e.} $t=2$. The resulting quantized network coded packets are then forwarded to the decoder node for recovery. 
In such scenario, we discuss robust recovery of messages and derive mathematical guarantees for local network coding coefficients.
One-step QNC is a hybrid method which consists of two stages of transmission, as explained in the following.

Assuming initial rest condition in the network, we have: $Y_e(1)=0,~\forall e \in \mathcal{E}$.
The messages, $X_v$, are available for transmission at $t=1$.
Therefore, the passed packets between the nodes and received at $t=2$ are only the raw quantized messages; \textit{i.e.} $\textbf{Q}_e(X_v)$'s:
\begin{equation}
Y_e(2)=\textbf{Q}_e(X_{tail(e)}),
\end{equation}
where $\textbf{Q}_e(\centerdot)$ is the quantizer operator, corresponding to edge $e$.\footnote{The design of $\textbf{Q}_e(\centerdot)$ depends on the value of $L$ and $C_e$.}
Denoting the quantization noise by $N_e(2)$, we have the following additive form:
\begin{equation}
Y_e(2)=X_{tail(e)}+N_e(2),~\forall e \in \mathcal{E}.
\end{equation}
These quantized version of neighboring node messages are used to calculate a random linear combinations of them.
Specifically, we define $P_v$, as the one-step linear combination of messages, at node $v$, according to:
\begin{equation}\label{Eq:ondeStepQNC}
P_v=\sum_{e' \in In(v)} \beta_{v,e'} Y_{e'}(2)+ \alpha_{v} X_v,~\forall v \in \mathcal{V}.
\end{equation}
In (\ref{Eq:ondeStepQNC}), the local network coding coefficients, $\beta_{v,e'}$ and $\alpha_{v}$, are uniformly and randomly picked from $\{-\kappa,+\kappa\}$, $\kappa>0$.
In order to prevent from over flow, we make sure that the normalization condition of Eq.~3 in \cite{naba1} is satisfied.

Some of these random linear combinations, $P_v$'s, are then transmitted to the decoder node via packet forwarding scheme.
Specifically, $P_v$ is sent to the decoder node with a probability of $\frac{m}{n}$, where $m$ is the number of received packets at the decoder node, by time $t$.
Representing the $i$'th received packet at the decoder node by $\underline{Z}_{\rm{tot}}(t)$, we have:
\begin{equation}
\{\underline{Z}_{\rm{tot}}(t)\}_{i}=P_v+N_e(3),~v \rightarrow i,
\end{equation}
where $v \rightarrow i$ means that $P_v$ is forwarded to decoder and corresponds to the $i$'th received packet.
Outgoing edge, $e \in Out(v)$ is also the edge on which $P_v$ is sent out of node $v$, and therefore $\textbf{Q}_e(\centerdot)$ is used for quantization of $P_v$.

Because of the linearity of QNC scenario, we can reformulate $\underline{Z}_{\rm{tot}}(t)$ as:
\begin{equation}\label{Eq:measEq}
\underline{Z}_{\rm{tot}}(t)=\Psi_{\rm{tot}}(t) \cdot \underline{X}+\underline{N}_{\rm{eff,tot}}(t),
\end{equation}
where $\Psi_{\rm{tot}}(t)$ and $\underline{N}_{\rm{eff,tot}}(t)$ are the total measurement matrix and the total effective noise vector, respectively, calculated according to:
\begin{equation}\label{Eq:defPsiTot}
\{\Psi_{\rm{tot}}(t)\}_{i,v}=
\left\{
\begin{array}{l l}
  \beta_{v',e'}  & ,~\scriptsize v' \rightarrow i, v \overset{e'}{\rightarrow} v', \\
  \alpha_{v'}  & ,~\scriptsize v' \rightarrow i, v'=v, \\
  0  &  ,~\mbox{otherwise} \\ \end{array} \right.
\end{equation}
\begin{equation}
\{\underline{N}_{\rm{eff,tot}}(t)\}_{i}=N_e(3)+\sum_{e' \in In(v)}\beta_{v,e'} ~N_{e'}(2),~v \rightarrow i. \label{Eq:defNeffTot}
\end{equation}
In (\ref{Eq:defPsiTot}), $v \overset{e'}{\rightarrow} v'$ denotes the existence of edge $e'$ from $v$ to $v'$.

Motivated by compressed sensing theory \cite{CS}, if we are able to recover the messages, $\underline{X}$, from fewer number of received packets than the number of messages, \textit{i.e.} $m <n$, then we have been able to perform an embedded non-adaptive compression of correlated messages.

Bayesian QNC was shown to be working better than packet forwarding by deducing the required number of received packets for robust recovery.
In the following, we present a theorem which offers robust recovery guarantee for one-step QNC scenario, as a simplification of (full) Bayesian QNC.
Specifically, it justifies the use of one-step QNC for transmission of near sparse Gaussian messages, in terms of the information content of received packets at the decoder; \textit{i.e.} random linear combinations, $P_v$'s.

\begin{theorem}\label{th:guarantee}
For the described one-step QNC scenario, with almost $k$-sparse messages, characterized as in section~\ref{sec:ProbDesc}, if 
\begin{equation}
{\kappa}= \sqrt{\frac{2n^2}{n+|\mathcal{E}|}}, \label{Eq:kappaCond}
\end{equation}
and,
\begin{equation}
m={O}\Big ( \frac{(1+\gamma)}{\mu^2 } \frac{k^3 n^2 \log(k n^\gamma)}{n+|\mathcal{E}|} \log(n)  \Big ),
\end{equation}
where $\gamma,\mu>0$, one can decode $\underline{X}$, from the measurements of the form $\Psi_{\rm{tot}}(t) \underline{X}$, to $\underline{\hat{X}}(t)$, such that:
\begin{equation}
\vectornormInf{\underline{X}-\underline{\hat{X}}(t)} \leq \mu \sigma_s, \label{Eq:errBound}
\end{equation}
with a probability exceeding $1-n^{-\gamma}$.
\footnote{A complete definition of $O(\centerdot)$ and $o(\centerdot)$ can be found in \cite{oNotation}.}
\end{theorem}

\begin{proof}
\footnote{Our proof uses a similar approach as in the proof of theorem~1 in \cite{baron2010bayesian}.}
To prove the theorem, we begin by finding probabilistic bounds on the $\ell_2$ and $\ell_\infty$ norms of $\underline{X}$.
Since $\phi$ is an orthonormal matrix, $\vectornorm{\underline{X}}=\vectornorm{\underline{S}}$, implying the same probabilistic $\ell_2$ norm of $\underline{X}$, as that of $\underline{S}$. Therefore, by using Eqs.~8,10 in \cite{baron2010bayesian}, we have:
\begin{eqnarray}
\textbf{P}( \vectornorm{\underline{X}}^2 < k \sigma^2_s) & = & {o}(n^{-\gamma}), \label{Eq:l2normbound1} \\
\textbf{P}(\vectornorm{\underline{X}}^2 > 2k \sigma^2_s + (n-k)\sigma^2_z ) & = & {o}(n^{-\gamma}). \label{Eq:l2normbound2}
\end{eqnarray}

To obtain a bound on $\ell_\infty$ norm of $\underline{X}$, we use Eq.~17 in \cite{baron2010bayesian}.
However, since their discussion is for a canonical case of $\phi$, we need to take care of non-diagonal $\phi$, in our case.
Basically, the orthonormal $\phi$ rotates the axes in $\underline{S}$ domain to $\underline{X}$ domain.
Considering an extreme case in which all $k$ non-zero elements of $\underline{S}$ are at their maximum possible limit (say $\sqrt{2 \log(k n^{\gamma})} \sigma_s $ from (17) in \cite{baron2010bayesian}), $\ell_2$ norm of $\underline{S}$ is $\sqrt{k}$ times that limit.
Now, consider $\phi$ rotates the axes such that we have an axis along the same direction as that extreme case of $\underline{S}$.
In such case, the $\ell_\infty$ norm of $\underline{X}$ corresponds to the $\ell_2$ norm of $\underline{S}$ in the aforementioned extreme case, implying:
\begin{equation}
\textbf{P}(\vectornormInf{\underline{X}} < \sqrt{2 k \log(k n^{\gamma})} \sigma_s) > 1-\frac{n^{-\gamma}}{2},
\end{equation}
and $\frac{\vectornormInf{\underline{X}}^2}{\vectornorm{\underline{X}}^2} \leq {2\log(k n^\gamma)}$, with overwhelming probability.

Assuming (without loss of generality) that $P_{v'}(2)$ is forwarded to decoder and corresponds to the $i$'th received packet, for $\{\Psi_{\rm{tot}}(t)\}_{iv}$'s, we have:
\begin{eqnarray}
\textbf{P}(\{\Psi_{\rm{tot}}(t)\}_{iv} = 0) &  = & \textbf{P}(v \neq v',\not{\exists} e':v \overset{e'}{\rightarrow} v') \nonumber \\
&  = & \textbf{P}(v \neq v') ~ \textbf{P}(\not{\exists} e':v \overset{e'}{\rightarrow} v'|v \neq v') \nonumber \\
&=& \Big( 1-\frac{1}{n} \Big) \Big( 1-\frac{(|\mathcal{E}|/n)!\binom{n-1}{|\mathcal{E}|/n}}{(|\mathcal{E}|/n)!\binom{n-2}{|\mathcal{E}|/n-1}} \Big ) \nonumber \\
&  = & \Big( 1-\frac{1}{n} \Big)\Big( 1-\frac{|\mathcal{E}|}{n(n-1)} \Big).
\end{eqnarray}
Since $+\kappa$ and $-\kappa$ are picked with the same probability for network coding coefficients, we have:
\begin{eqnarray}
\textbf{P}(\{\Psi_{\rm{tot}}(t)\}_{iv} = +\kappa) &=& \textbf{P}(\{\Psi_{\rm{tot}}(t)\}_{iv} = -\kappa) \nonumber \\
&=& \frac{1}{2}-\frac{1}{2}\textbf{P}(\{\Psi_{\rm{tot}}(t)\}_{iv} = 0) \nonumber \\
&=& \frac{1}{2}-\frac{1}{2}\Big( 1-\frac{1}{n} \Big)\Big( 1-\frac{|\mathcal{E}|}{n(n-1)} \Big) \nonumber \\
&=& \frac{n+|\mathcal{E}|}{2n^2} \nonumber \\
&=& \frac{1}{\kappa^2}.
\end{eqnarray}
For $\forall i,v,~1 \leq i \leq m, 1 \leq v \leq n$, this implies:
\begin{eqnarray}
\textbf{E}[\{\Psi_{\rm{tot}}(t)\}_{iv}] &=& 0, \nonumber \\
\textbf{E}[\{\Psi_{\rm{tot}}(t)\}^2_{iv}] &=& 1, \nonumber \\
\textbf{E}[\{\Psi_{\rm{tot}}(t)\}^4_{iv}] &=& \kappa^2,
\end{eqnarray}
and therefore:
\begin{equation}
\kappa^2 \cdot {2\ln(k n^\gamma)} = \frac{4n^2 \ln(k n^\gamma)}{n+|\mathcal{E}|}.
\end{equation}
By applying theorem~1 in \cite{wang2007distributed}, using a similar reasoning as in the proof of theorem~1 in \cite{baron2010bayesian}, and choosing
\begin{equation}
\epsilon=\frac{\mu}{\sqrt{2k}},
\end{equation}
%
we can finish the proof of our theorem.
\end{proof}
The following can also be established from theorem~\ref{th:guarantee}.
\begin{theorem}\label{th:guarantee2}
For the described one-step QNC scenario, if (\ref{Eq:kappaCond}) is satisfied, and,
\begin{equation}
|\mathcal{E}| = O(n^2),
\end{equation}
\begin{equation}
m={O} \Big ( \frac{(1+\gamma)}{\mu^2 } k^3 \log(k n^\gamma)\log(n)  \Big ),
\end{equation}
where $\gamma,\mu>0$, then (\ref{Eq:errBound}) holds with a probability, exceeding $1-n^{-\gamma}$.
\end{theorem}

Theorem~\ref{th:guarantee2} states that in a (relatively) highly connected network deployment (high number of edges), we need smaller order of received packets, $m$, as the order of messages, $n$, to be able to recover the messages.
This saving implies an embedded distributed compression, which is achieved without adapting the local network coding coefficients (We can pick $\alpha_{v}$'s and $\beta_{v,e}$'s from $\{-1,+1\}$ and then scale the received packets by an appropriate constant result in the same $\Psi_{\rm{tot}}(t)$ as picking them from $\{-\kappa,+\kappa\}$).
In other words, the total measurement matrix, resulting from one-step QNC scenario has similar characteristics as an appropriate measurement matrix for Bayesian compressed sensing.
Therefore, a similar recovery guarantee can be offered for our case of one-step QNC.

\section{Simulation Results}
\label{sec:simRes}

\begin{figure*}[t!]
\centering
\subfigure[$|\mathcal{E}|=400$ edges]{
\resizebox{.78\textwidth}{!}{
\includegraphics{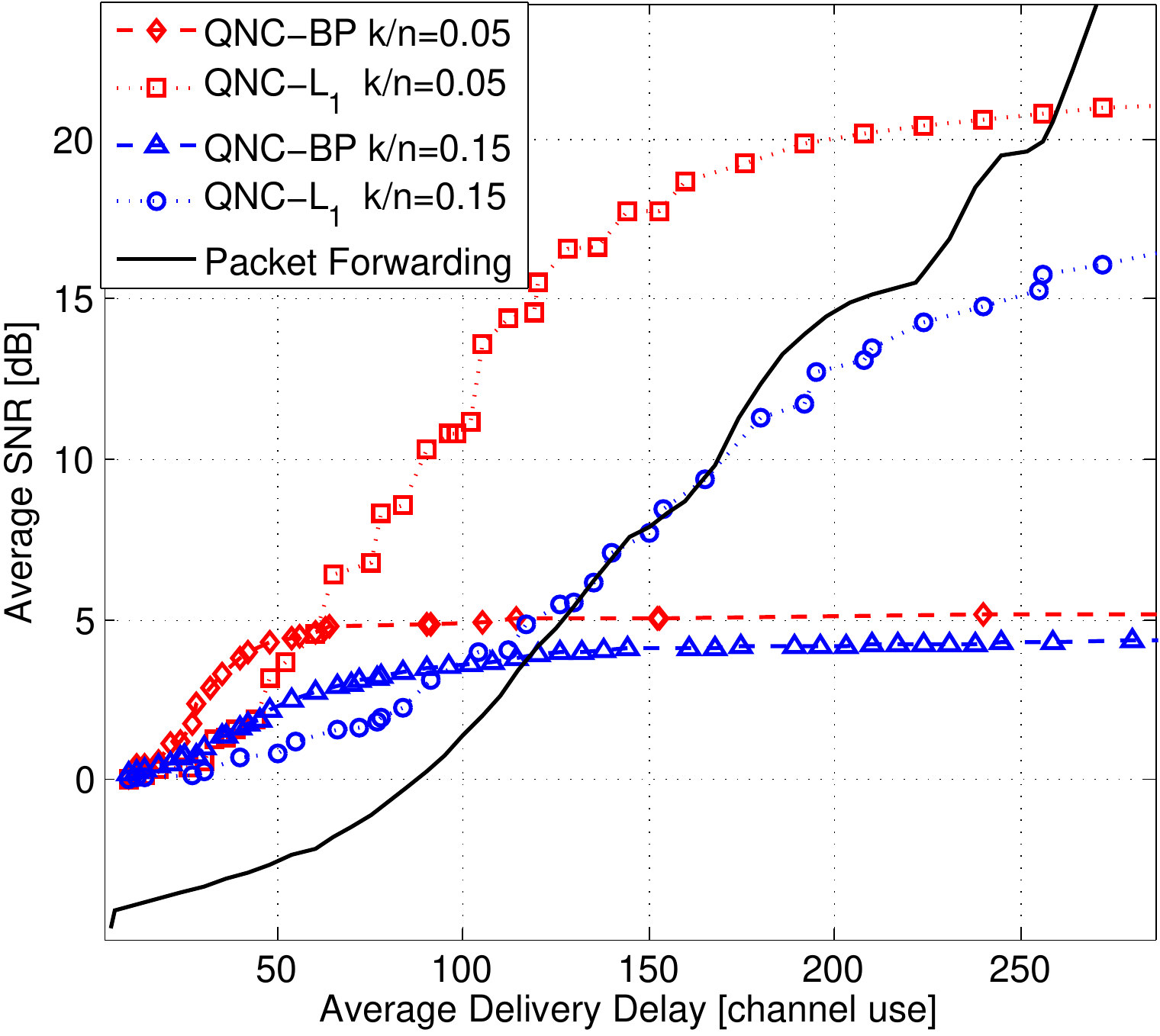}}
\label{fig:4}
} \\
\subfigure[$|\mathcal{E}|=800$ edges]{
\resizebox{.78\textwidth}{!}{
\includegraphics{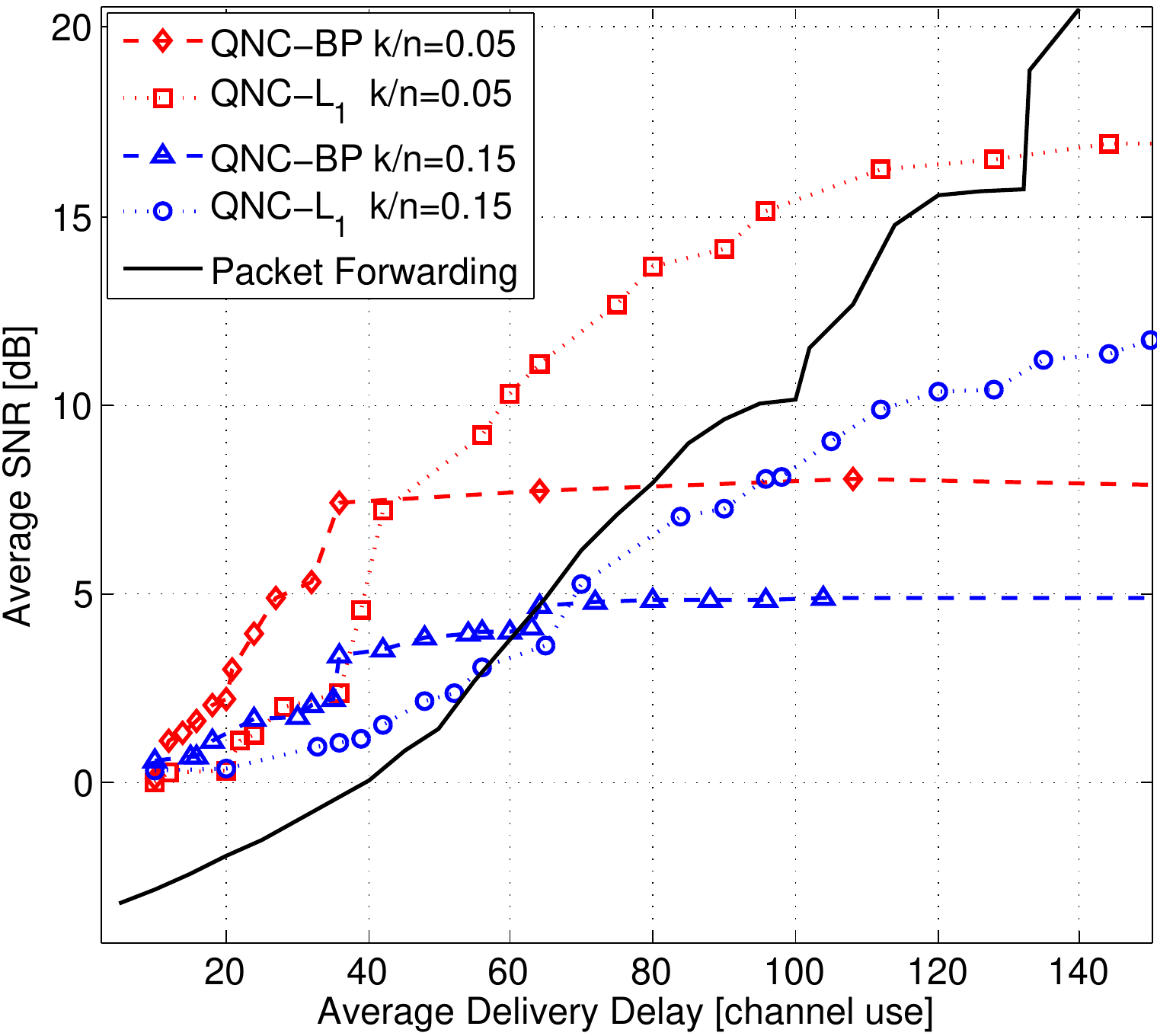}}
\label{fig:8}
} 
\caption{Average SNR versus average delivery delay of one-step QNC and packet forwarding for different sparsity factors, $\frac{k}{n}$, and different number of edges\label{fig:subfigureExample}.
}
\end{figure*}

In this section, we compare the performance of one-step QNC with that of routing based packet forwarding.
To do so, different deployments of a sensor network with $n=100$ nodes and $|\mathcal{E}|=400,800$ uniformly distributed edges are generated randomly.
The edges are directed and can maintain a lossless interference-free communication of $C_e=1$ bit per use.
Messages are also randomly generated using a random $\phi$ according to the model of Eq.~\ref{Eq:modelMess}, for $\sigma^2_s=1$ and $\sigma^2_z=0.01$, and different sparsity factors: $\frac{k}{n}=0.05,0.15$.

For each deployment, one-step QNC is run by using a uniform quantizer of appropriate step size (depending on the value of block length, $L$ and putting the dynamic range between $-4 \sigma_s$ and $+4 \sigma_s$).
In a progressive manner, we use a routing based packet forwarding to deliver all $P_v$'s to the decoder node.
Then, at each time, $t$, the received packets up to $t$ are used to recover the messages, which lets us obtain different quality and delay performance points.
Specifically, we use both Belief Propagation (BP) based minimum mean square error decoding \cite{naba3} and $\ell_1$-min decoding \cite{naba1} to recover messages, from the received $P_v$'s up to $t$

In a different scenario, for each deployment, packet forwarding via optimal route is run and messages are delivered to the decoder node.
Each message, $X_v$, is quantized at its source by using a similar uniform quantizer as used in one-step QNC scenario.
The delivered quantized messages up to time $t$ are taken care for calculating the corresponding signal to noise ratio.
Moreover, the routes from nodes to the decoder node are calculated by using Dijkstra algorithm \cite{dijkstra1959note}. 

To obtain a quality-delay diagram of the performance, we find the best block length, $L$, for each SNR value, corresponding to each one-step QNC and packet forwarding scenario.
The resulting delivery delays are then averaged over different deployments of network.
The average SNR (quality measure) is depicted versus the average delivery delay (cost measure), for $400$ and $800$ edges in Fig.~\ref{fig:4} and Fig.~\ref{fig:8}, respectively.

As it was expected from the mathematical derivations of section~\ref{sec:QNC}, the average delivery delay of one-step QNC for a given SNR is less than that of packet forwarding, in a wide range of SNR values.
Smaller a sparsity factor, $\frac{k}{n}$, meaning higher level of correlation between messages, results in a better performance.
Moreover, as shown in Figs.~\ref{fig:4},\ref{fig:8}, larger number of edges (higher edge densities) increases the performance gap between the one-step QNC and packet forwarding scenarios.
Although the resulting performance of one-step QNC may not look promising for all SNR values, as it was illustrated in \cite{naba3}, using (full) QNC scenario helps us get a significant improvement for all SNR values.

\section{Conclusions}
\label{sec:Conclusions}
Theoretical motivations behind Bayesian QNC are discussed by deriving mathematical guarantees for robust recovery of near sparse Gaussian messages, in a simple QNC scenario, called one-step QNC.
Our derived conditions for robust recovery, which bounds the $\ell_{\infty}$ norm of error, show that one-step QNC requires smaller order of received packets at the decoder than that of messages.
This implies an embedded distributed compression, while there is no need to adapt the encoding with the correlation model of messages (\textit{i.e.} sparsifying transform, $\phi$).
However, more mathematical works are still necessary to analyze the (full) Bayesian QNC scenario to obtain tighter bounds for robust recovery of messages.

\bibliographystyle{IEEEbib}
\bibliography{Ref_arXiv}
\end{document}